\documentclass{llncs}

\usepackage{verbatim}
\usepackage{amssymb}
\usepackage{graphicx}
\usepackage{enumerate}
\usepackage{times}
\usepackage{url}
\usepackage{wrapfig}
\usepackage{color}

\let\doendproof\endproof
\renewcommand\endproof{~\hfill\qed\doendproof}

\newcommand{\emb}{\mathcal{E}}
\newcommand{\ps}{\mathcal{P}}
\newcommand{\TT}{\mathcal{T}}

\newenvironment{sketch}{\noindent{\it Proof Sketch.}}{\mbox{}\hfill\qed\par}

\title{On the Book Thickness of 1-Planar Graphs
\thanks{Supported in part by the Deutsche Forschungsgemeinschaft (DFG), grant Br835/18-1.}
}

\author{
Md.~Jawaherul~Alam\inst{1}
\and
Franz J. Brandenburg\inst{2}
\and
Stephen~G.~Kobourov\inst{1}
}

\institute{
    Department of Computer Science, University of Arizona, USA\\
    {\tt \{mjalam, kobourov\}@cs.arizona.edu}
\and
    University of Passau, 94030 Passau, Germany\\
    {\tt brandenb@informatik.uni-passau.de}
}

\begin{document}

\maketitle

\begin{abstract}
In a \textit{book embedding} of a graph $G$, the vertices of $G$ are
 placed in order along a straight-line called \textit{spine} of the book, and the edges of $G$ are drawn on a  set of half-planes, called the \textit{pages} of the book, such that
two edges drawn on a page do not cross each other.
The minimum number of pages in which a graph can be embedded is called the \textit{book-thickness} or the \textit{page-number} of the graph.
 It is known that every planar graph has a book embedding on at most four pages.
 Here we investigate the book-embeddings of \textit{1-planar} graphs.
A graph is \textit{1-planar} if it can be drawn in the plane such that each edge is crossed at most once.
 We prove that every 1-planar graph has a book embedding
 on at most 16 pages and every 3-connected 1-planar graph has a book embedding
 on at most 12 pages.
 The drawings can be computed in linear time from any given 1-planar embedding of the graph.
\end{abstract}

\section{Introduction}

Graph embeddings and linear layouts of graphs play an important role
in graph drawing, parallel processing, matrix computation, VLSI
design, and permutation sorting. A linear layout prescribes the
order in which the vertices are processed and the embedding of the
edges reveals structural properties of the given graph. A particular
example is a book embedding in which the edges are assigned to pages
such that edges in the same page nest and do not cross.
Equivalently, the vertices are visited in the linear order and the
edges are processed in stacks. The concept of a book embedding of a
graph was introduced by Ollmann~\cite{ollmann1973book} and by Kainen~\cite{Kainen74}
and can be formalized as follows.
A $k$-page book embedding  of a
graph $G=(V, E)$ is defined by a linear order of the vertices of $G$ and
a partition of the edges into $k$ sets $E_1, \dots E_k$, so that the vertices of $G$
are placed on a line in the given order and edges in $E_i$ are drawn on page $i$
(typically with circular arcs), so that no two edges on the same page cross.
The book thickness of the graph $G$ is the smallest number of pages
needed, also known as stack number or page number.

The book thickness of planar graphs has been studied for over 40 years.
 Bernhart and Kainen~\cite{bk-btg-79} characterized the graphs with book
 thickness one as the outerplanar graphs and the graphs with book thickness
 two as the sub-Hamiltonian planar graphs. Deciding whether a general planar
 graph has book thickness two is NP-hard~\cite{chung1987embedding}.
 It is known that planar graphs require 3 pages and a series of improvements
 brought down the upper bound from 9~\cite{buss1984pagenumber},
 to 7~\cite{heath1984embedding}, and 4~\cite{Yanna89}.
Although in an earlier version of his 1989 paper Yannakakis in 1986~\cite{YannaSTOC86}
 claimed 4 pages are necessary, and later Dujmovic and Wood in 2007~\cite{dujmovic2007graph}
 also conjectured the same lower bound, there is still no conclusive evidence that this is indeed the case.

More recently there has been a greater interest in studying non-planar graphs which
extend planar graphs by restrictions on crossings. A particular
example are {\em 1-planar graphs} which can be drawn in the plane with at
most one crossing per edge. Such graphs were first defined by Ringel in the context of simultaneously drawing a planar graph
and its dual~\cite{ringel-65}. In many respects, 1-planar
graphs 
generalize planar graphs. There are 1-planar
embeddings as witnesses of 1-planarity, in which the crossings are
treated as special vertices of degree four, and which then result in planarizations. 
Like $n$-vertex planar graphs which have at most $3n-6$ edges,
$n$-vertex 1-planar graphs have at most $4n-8$ edges~\cite{PT97}.
 Both planar and 1-planar 3-connected
graphs admit straight-line drawings in $O(n^2)$ area (with the exception of one edge in the outer face for the densest 1-planar graphs)~\cite{ABK13}.
However, there is a major difference in the complexity of the
recognition of planar and 1-planar graphs, which can be done in linear time
for planar graphs while it is $NP$-hard for 1-planar graphs~\cite{GB07,KorzhikMohar13}. 
On the other hand, there is a cubic time recognition algorithm for hole-free map graphs~\cite {CGP06},
 which for 3-connected graphs coincide with \textit{planar-maximal}
 1-planar graphs (i.e., where no edge can be added without creating more crossing).



In this paper, we address the problem of book embedding of 1-planar graphs.
 Recently Bekos \textit{et al.}~\cite{BBKR15} gave a constant upper bound of 39 on the book
 thickness of 1-planar graphs.
Here we prove that 1-planar graphs have book thickness at most 16 and 3-connected 1-planar
 graphs have book thickness at most 12.
 If the planar skeleton is Hamiltonian, then four pages suffice, and we have found 1-planar
 graphs which need four pages.

\section{Preliminaries}


 A \emph{drawing} of a graph $G$ is a mapping of $G$ into the plane such that vertices
 are mapped to distinct points and edges are Jordan arcs between their endpoints.
 A drawing is \emph{planar} if the edges do not cross and it is
 \emph{1-planar} if each edge is crossed at most once.
Hence in a 1-planar drawing the crossing edges come in pairs.
For example, $K_5$ and $K_6$ are 1-planar graphs.
An \emph{embedding} of a graph is planar (resp. 1-planar) if it admits a
 planar (resp. 1-planar) drawing. An embedding specifies
 the \emph{faces}, which are topologically connected regions.
 The unbounded face is the \emph{outer face}.
Accordingly, a 1-\emph{planar embedding} $\emb(G)$ specifies the
faces in a 1-planar drawing of $G$ including the outer face. A
1-planar embedding is a witness for 1-planarity. In particular,
$\emb(G)$ describes the pairs of crossing edges, the faces where
the edges cross, and the \emph{planar} edges.

Augment a given 1-planar embedding $\emb(G)$  by adding as many edges to $\emb(G)$ as possible
 so that $G$ remains a simple graph and the newly added edges are planar in $\emb(G)$. We call such
 an embedding a \textit{planar-maximal} embedding of $G$ and the operation \textit{planar-maximal
 augmentation}.
Then each pair of crossing edges is augmented to a $K_4$.
 The \emph{planar skeleton} $\ps(\emb(G))$ consists of the planar edges of a
 planar-maximal augmentation. It is a planar embedded graph, since all pairs of crossing edges are
 omitted. Note that the planar augmentation and the planar skeleton are defined for an embedding,
 not for a graph.

The \emph{normal form} for an embedded 3-connected 1-planar graph $\emb(G)$ is obtained by first
 adding the four planar edges to form a $K_4$ for each pair of crossing edges while routing them close
 to the crossing edges and then removing old duplicate edges if necessary. Such an embedding of a
 3-connected 1-planar graph is a {normal embedding} of it. A \textit{normal planar-maximal augmentation}
 for an embedded 3-connected 1-planar graph is obtained by first finding a normal form of the embedding
 and then by a planar-maximal augmentation. 

Given a 1-planar embedding $\emb(G)$, the normal planar-maximal augmentation of $\emb(G)$
 can be computed in linear time~\cite{ABK13}. We say that an embedded 3-connected 1-planar graph
 is a \textit{normal planar maximal} 1-planar graph if a normal planar maximal augmentation of the
 graph yields the same graph. In a 3-connected normal planar-maximal 1-planar graph, each pair of
 crossing edges $(a,c)$ and $(b,d)$ crosses each other either inside or outside the boundary of the
 quadrangle $abcd$ of the planar edges, and these define the so-called \textit{augmented X-} and
 \textit{augmented B-configurations}~\cite{ABK13}.


%

For a 3-connected 1-planar graph $G$, Alam \textit{et al.}~\cite{ABK13} proved the following:

\begin{lemma}~\cite{ABK13}
\label{lem:embedding} Let $G$ be a 3-connected 1-planar graph with a 1-planar embedding $\emb(G)$.
 Then the normal planar-maximal augmentation of $\emb(G)$ gives a planar-maximal 1-planar embedding
 $\emb(G^*)$ of a supergraph $G^*$ of $G$ so that $\emb(G^*)$ contains at most one augmented
 B-configuration in the outer face and each augmented X-configuration in $\emb(G^*)$ contains no vertex
 inside its skeleton.
\end{lemma}

\section{Book Embeddings of 3-Connected 1-Planar Graphs}
\label{sec:3-connected}

If a graph can be embedded in a given number of pages, the same is true for its subgraphs. Given an
 embedded 3-connected 1-planar graph $G$, we therefore assume that $G$ is a normal planar
 maximal 1-planar graph. Lemma~\ref{lem:embedding} implies that the planar skeleton of a normal
 planar maximal 3-connected 1-planar graph $G$ contains only triangular and quadrangular faces.
 Furthermore if we remove exactly one crossing edge (arbitrarily) from each pair of crossing edges in
 $G$, then the resulting graph is a maximal planar graph.

Our algorithm uses the a ``peeling technique'' similar to Yannakakis~\cite{Yanna89} and iteratively
 removes the vertices on the outer cycle of the planar skeleton $\ps(\emb(G))$ of $G$. This partitions the
 vertices of $G$ into \textit{levels} according to their ``distance'' from the outer face of the planar
 skeleton $\ps(\emb(G))$. Vertices on the outer face of $\ps(\emb(G))$ are at level 0. Deleting these
 vertices from $\ps(\emb(G))$ yields the level 1 graph; the vertices that lie now on the outer face are at
 level 1. In general,
 the level $t$ graph is obtained by deleting all vertices at levels less than $t$; the vertices that lie on
 the outer face of this graph are at level $t$. The edges of $G$ (including the crossing edges) are
 partitioned into \textit{level edges} at level $i$, edges that connect vertices at the same level $i$,
 and \textit{binding edges}, edges that connect vertices at different levels. The fact that a level $i$
 vertex is not on the outer face after deleting the first $i-2$ levels implies that every level $i$ vertex
 lies in the interior of some cycle composed of level $i-1$ vertices. This means in particular that a
 level $i$ vertex cannot be adjacent to a level $j$ vertex with $j < i-1$ and binding edges connect
 only consecutive levels.

Similar to Yannakakis~\cite{Yanna89} we first place level 0 vertices in the clockwise order (cw-order)
 as they appear on the outer cycle, assigning the edges on the outer cycle on the same page.
 Then we place the level 1 vertices and assign the following edges to some pages:
 (i) the level edges of each cycle on the outer boundary of the level 1 graph
 (ii) the binding edges between levels 0 and 1
 (iii) the crossing edges either at level $0$ or binding between level 0 and 1.
Level 1 vertices are placed in such a way that the vertices on each level 1 cycle are
 in the counterclockwise order (ccw-order) around the cycle.
 Now the rest of the graph is in the interior of level 1 cycles. The algorithm takes each
 level 1 cycle in turn and lays out its interior in a similar way.

We therefore next consider a \textit{2-level subgraph} $H$ of $G$ defined as follows. The vertices of $H$
 are the vertices on a level $i$ cycle $C_i$ and all the level $i+1$ vertices $V_{i+1}$ interior to $C_i$.
 The edges of $H$ are all the planar and crossing edges inside the region between $C_i$ and the outer
 boundaries of all the level $i+1$ components inside $C_i$ (including the edges on $C_i$ and the level
 $i+1$ boundaries). Fig.~\ref{fig:2-level} shows a 2-level subgraph inside a cycle $C_i=AB\ldots Z$.
 We denote this 2-level subgraph of $H$ inside $C_i$ as $H(C_i)$. We assume that $C_i$ has already
 been embedded where the vertices of $C_i$ are placed in the cw (or ccw, resp.)
 order around $C_i$. We then extend this embedding to a book embedding of $H(C_i)$, by
 placing the remaining vertices of $H(C_i)$ and assign the remaining edges of $H(C_i)$ to seven pages.
The book embedding of $G$ is obtained by repeatedly computing the book embeddings of
 $H(C_i)$ and reusing the same seven pages for all odd (even) $i$.



\subsection{Drawing 2-Level Subgraphs}

In this section we prove the following lemma.

\begin{lemma}
\label{lem:2-level}
Let $H(C_i)$ be a 2-level subgraph of $G$ inside
a level $i$ cycle $C_i$. Then there exists a book embedding
$\Gamma$ of $H(C_i)$ on seven pages where the vertices of
 $C_i$ are placed in the cw (or ccw) order around $C_i$.
\end{lemma}

\begin{figure}[t]
\centering
\includegraphics[width=0.8\textwidth]{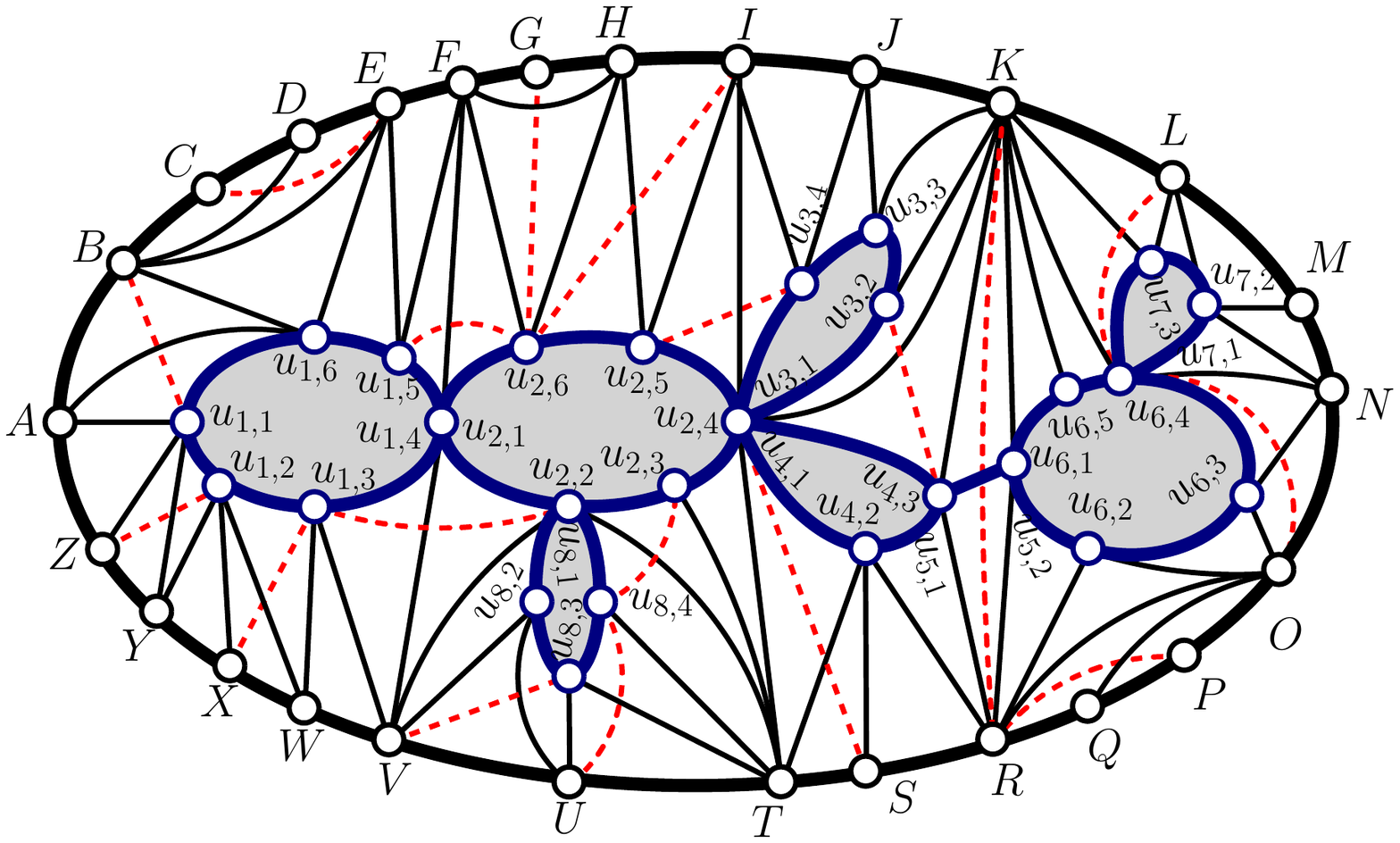}
\caption{A 2-level subgraph $H(C_i)$ of $G$ inside the level-$i$ cycle $C_i=AB\ldots Z$, which is
 drawn with thick black edges. The outer boundary of the level $i+1$ component is drawn with thick
 blue edges. The red dashed edges are the crossing edges taken in the set $X$.}
\label{fig:2-level}
\end{figure}

We give a construction of a book embedding where the vertices of $C_i$ are placed in the cw-order
 (for ccw-order we flip the embedding of $H(C_i)$).
Let $v_1, \ldots, v_t$ be the vertices of $C_i$ in the cw-order around $C_i$. For the remaining part
 of this section, we call the vertices on $C_i$ as the \textit{outer vertices} and the level $i+1$ vertices
 of $H(C_i)$ as the \textit{inner vertices}.
We first obtain a planar graph $H'$ from $H(C_i)$ by removing exactly one edge from each
 pair $\langle (a,b), (c,d)\rangle$ of crossing edges. Let $X$ be the set of crossing edges that we remove.
 From each crossing edge pair $\langle (a,b), (c,d)\rangle$, we take one edge
 to be in $X$ as follows; see Fig.~\ref{fig:2-level}.

\noindent
\textbf{Case S1.} If both $(a,b)$, $(c,d)$ are level edges at level $i$, then we take the edge adjacent to the
  vertex farthest from $v_1$ in cw-order on $C_i$ to be in $X$. In particular if the two level $i$ edges
 forming the crossing pair are $(v_p, v_r)$, $(v_q, v_s)$ with $p<q<r<s$, then we take the edge
 $(v_q,v_s)$ to $X$; for example we take the edge $(C, E)$ to $X$ in Fig.~\ref{fig:2-level}.


\noindent
\textbf{Case S2.} If both $(a,b)$, $(c,d)$ are binding edges, then we again choose the edge adjacent to
 the vertex farthest from $v_1$ in cw-order on $C_i$ to be in $X$. In particular, if the two binding edges
 forming the crossing pair are $(v_p, u)$, $(v_q, u')$, where $p<q$ and $u,u'$ are level $i+1$ vertices,
 then we take the edge $(v_q,u')$ to be in $X$; for example we take the edge $(I, u_{2,6})$ to $X$ in
 Fig.~\ref{fig:2-level}.

\noindent
\textbf{Case S3.} If one of $(a,b)$, $(c,d)$ is a level edge at level $i$, and the other is a binding edge,
 then we choose the binding edge to be in $X$;  for example we take the edge $(G, u_{2,6})$ to $X$ in
 Fig.~\ref{fig:2-level}.

\noindent
\textbf{Case S4.} If one of $(a,b)$, $(c,d)$ is a level edge at level $i+1$, and the other is a binding edge,
 then we choose the level edge at level $i+1$ to be in $X$; for example we take the edge
 $(u_{1,5}, u_{2,6})$ to $X$ in Fig.~\ref{fig:2-level}.

\noindent
\textbf{Case S5.} If one of $(a,b)$, $(c,d)$ is a level edge at level $i$, and the other is at level $i+1$, then
 we choose the level edge at level $i$ to be in $X$; for example we take the edge $(K, R)$ to $X$ in
 Fig.~\ref{fig:2-level}.

Note that due to the construction of $H(C_i)$, the pair of crossing edges cannot both be level edges at
 level $i+1$. Thus the above cases account for all possible pairs.

We then use the algorithm by Yannakakis~\cite{Yanna89} to obtain a book embedding of $H'$ on three
 pages, and we add the crossing edge from $X$ on four additional pages so that no two edges assigned
 to the same page cross each other on that page. Before we describe how
 to add the crossing edges, we describe the 3-page embedding of $H'$.
 Denote the three pages as
 $p_1$, $p_2$ and $p_3$, and denote the four additional pages for the crossing edges as $c_1$,
 $c_2$, $c_3$ and $c_4$.
 Denote by $D$ the subgraph of $H(C_i)$ induced by the vertices at level $i+1$. Assume without loss
 of generality that $D$ induces a connected graph, since otherwise each connected component of $D$
 would be inside a different cycle induced by the vertices of $C_i$ and these can be handled separately.
 By construction then, each biconnected block of $D$ is a simple cycle (i.e., $D$ is a {\em cactus graph}).
 Let $B_1, \ldots, B_s$ be these blocks of $D$ and let $\TT$ be the block-cut tree for $D$.
 We now show how we place all the level $i+1$ vertices and assign the edges in $H'$ and in $X$
 to the seven pages $p_1$, $p_2$, $p_3$, $c_1$, $c_2$, $c_3$ and $c_4$.

\subsubsection{Placement of Vertices}

We say that a vertex $u$ \textit{sees} an edge $(v, w)$ if $uvw$ forms a triangular face in $H'$.
 We say that an outer vertex \textit{sees} a block $B_j$ of $D$ if it sees an edge of the block.
Consider the triangular inner face containing the edge $(v_1, v_t)$ of $C_i$. The third node of
 this face $u_{1,1}$ is called the first inner node and assume the block $B_1$ containing $u_{1,1}$
 is the \textit{first block}\footnote{Assume $u_{1,1}$ is in a unique block;
 otherwise take as $B_1$ a block that has $u_{11}$ and is seen by $v_1$.}. Then consider
 the block-cut tree $\TT$ as a rooted tree by taking $B_1$ as its root.

For each block $B_j$ of $D$, define the \textit{leader} of $B_j$ to
be the first vertex of $B_j$
 in any path from $u_{1,1}$ to any vertex of $B_j$. Thus, the leader of the root block $B_1$
 is $u_{1,1}$; for any other block $B_j$, the leader of $B_j$ is the common vertex between $B_j$
 and its parent in $\TT$. Although an inner vertex of $H'$ (in particular a cutpoint of $D$) may
 belong to more than one block, we assign each inner vertex $u$ to a unique block  by assigning
 it to the \textit{highest} (i.e., closest to the root) block in the tree $\TT$ that contains it. Thus, the root
 block $B_1$ of $\TT$ is assigned all its vertices; each remaining block is assigned all its vertices
 except its leader. The \textit{dominator} of a block $B_j$ is the first outer vertex (in the order
 $v_1, \ldots, v_t$) adjacent to a vertex assigned to $B_j$. 

We first place the outer vertices
 in the order $v_1, \ldots, v_t$  in $\Gamma$.
Next we place the inner vertices
 in between these outer vertices using the 
 vertex placement order in~\cite{Yanna89}, which we describe here. The inner vertices assigned to
 each block $B_j$ are placed right after the outer node $v_k$ that
 dominates $B_j$ (i.e., between $v_k$ and $v_{k+1}$). If an unique block $B_j$ is dominated by $v_k$,
 then its vertices are placed between $v_k$ and $v_{k+1}$ in the ccw-order around its
 boundary.
 If more than one block has a common dominator $v_k$, this set $S$ of blocks forms a directed path
 in $\TT$. 
 Vertices in these blocks are placed between $v_k$ and $v_{k+1}$ using one of two methods.
 In the \textit{nested method}, the vertices are placed in the order they are first
 encountered while traversing the boundary of the subgraph induced by the blocks in $S$ in
 ccw-order, starting with the leader of the highest block in $S$.
 In the \textit{consecutive method}, the vertices assigned to each block are placed consecutively in
 ccw-order around its boundary; the blocks are ordered one after the other in top-down
 order of $\TT$: first the vertices assigned to the highest block, then the ones assigned to its child,
 and so on. For the following description, assume that we follow the consecutive method, 
 (the algorithm is analogous with the nested method). We thus obtain the ordering of the
 vertices of $H'$ for the book embedding $\Gamma$; see Fig.~\ref{fig:embedding-planar}.

\begin{figure}[h]
\centering
\includegraphics[angle=90,width=\textwidth]{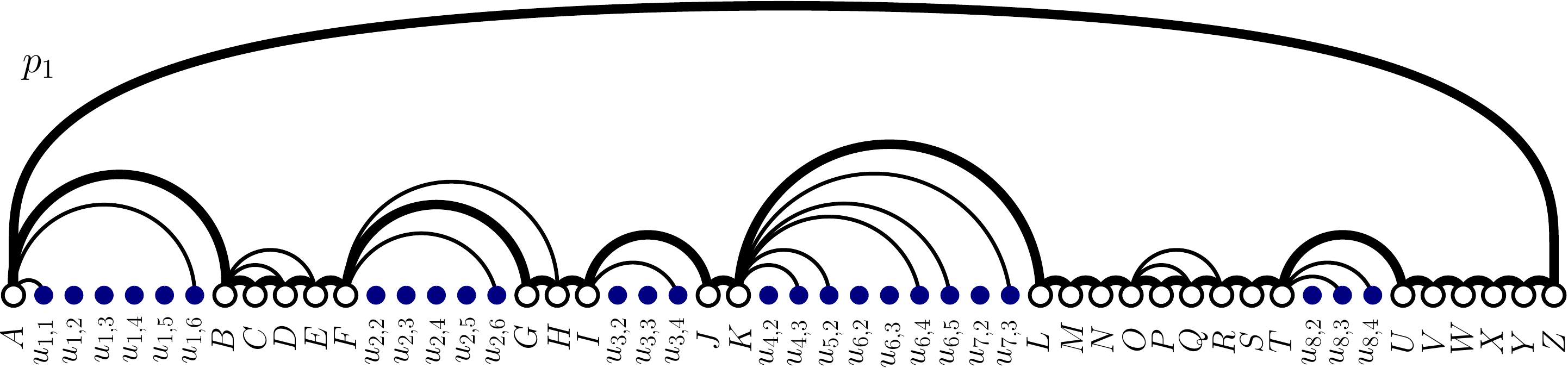}
\includegraphics[angle=90,width=\textwidth]{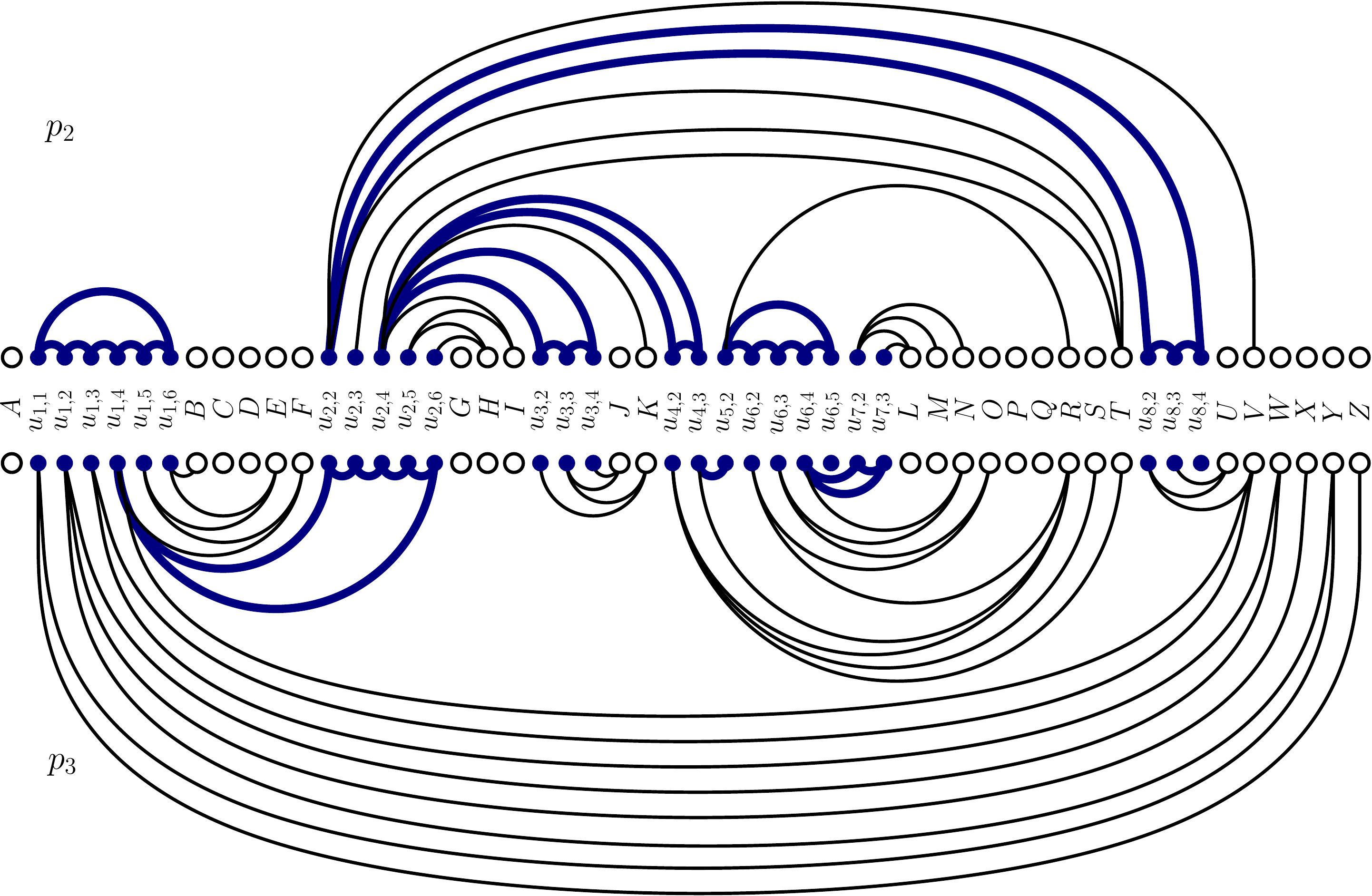}
\caption{Book embedding of the planar edges in $H'$ for the 2-level graph $H(C_i)$ in Fig.~\ref{fig:2-level}
 on the three pages $p_1$, $p_2$, $p_3$.}
\label{fig:embedding-planar}
\end{figure}

\subsubsection{Assigning Edges in $H'$ to Pages}

Next we assign the edges of $H(C_i)$ to the seven pages. For a vertex $v$ of $H(C_i)$, let $\Gamma(v)$
 denote its rank in the ordering of $\Gamma$. Consider two edge $(a,b)$, $(c,d)$ of $H(C_i)$ with
 $\Gamma(a)<\Gamma(b)$ and $\Gamma(c)<\Gamma(d)$. We say that there is a conflict between these
 two edges in $\Gamma$ if $\Gamma(a)< \Gamma(c) < \Gamma(b)< \Gamma(d)$ or
 $\Gamma(c) < \Gamma(a) < \Gamma(d) < \Gamma(b)$. We now assign the edges of $H(C_i)$ on seven
 pages such that there is no conflict between any two edges assigned to the same page.

First we assign the edges of $H'$ to the three pages $p_1$, $p_2$ and $p_3$. In order to see this
 assignment of edges to pages, consider $H'$ as a directed (acyclic) graph by taking
 the following orientation of edges.
The level edges $(v_p, v_q)$ at level $i$, with $p<q$ are oriented from $v_p$ to $v_q$ (including the
 edge $(v_1, v_t)$, which is oriented from $v_1$ to $v_t$).
 On the other hand, each inner cycle is traversed in ccw-order,
 starting from its leader and the edges are oriented accordingly, with the exception of the final edge,
 which is oriented away from the leader.
Each binding edge is oriented from the inner vertex to the outer
 vertex. The orientation of edges along with the placement of the vertices in $\Gamma$
 partitions all the edges of $H'$ in two types:  \textit{forward} edges have
 sources placed before their sinks in $\Gamma$ (the edge orientation is forward); the remaining
 edges are \textit{backward} edges (the edge orientation is backward).

Consider an assignment of the blocks of $D$ to the pages $p_2$ or $p_3$. The root block is assigned
 to $p_2$.
 In the nested method, for each non-root block $B_i$, if $B_i$ has a
 different dominator than its parent then it is assigned to the opposite page ($p_2$ or $p_3$) than that of
 its parent, otherwise it is assigned to the same page as its parent. In the consecutive method each
 non-root block $B_i$ is assigned a different page ($p_2$ or $p_3$) than that of its parent. Again we use
 the consecutive method for illustration. We assign all the edges of $H'$ to the three pages $p_1$, $p_2$,
 $p_3$ as follows; also see Fig.~\ref{fig:embedding-planar}.

\begin{itemize}
 \item The edges of $C_i$ and all the backward binding edges of $H'$ is assigned to page $p_1$.
 (These are the only edges of $H'$ assigned to $p_1$.)

 \item For each block $B_j$ the level edges in $B_j$ are assigned to the page that the block itself is assigned to.

 \item  For each forward binding edge $e=(u, v)$, where $v$ is on $C_i$ and $u$ is on some block
 $B_j$,  edge $e$ is assigned to page $p_2$ or $p_3$, opposite to the one assigned to block
 $B_j$.

\end{itemize}

This assignment of edges of $H'$ creates no edge conflicts in any of the three pages~\cite{Yanna89}.

\subsubsection{Assigning Edges in $X$ to Pages}

We now assign the edges in $X$ to the four pages $c_1$, $c_2$, $c_3$ and $c_4$.
 We consider the following cases of a crossing edge $(a,b)$ in $X$.

\noindent
\textbf{Case D1: $(a,b)$ is a binding edge}. 
A binding edge in $X$ is called \textit{forbidden} for some block $B_j$ if it is between two vertices $d$ and
 $v_{k+1}$, where $d$ is the leader of $B_j$, $v_k$ is the dominator of $B_j$, $v_{k+1}$ is the outer
 vertex just after $v_k$ and $v_k$ is not the dominator of any child block of $B_j$ in $\TT$.
We assign a binding edge $(a,b)$ to page $c_1$ if it is not forbidden for some block; see
 Fig.~\ref{fig:embedding-crossing}; otherwise we assign it to either page $c_3$ or page $c_4$ in Case D4.

\noindent
\textbf{Case D2: $(a,b)$ is a level $i$ edge}. In this case we assign $(a,b)$
 to page $c_1$; see Fig.~\ref{fig:embedding-crossing}.

\begin{figure}[t]
\centering
\includegraphics[angle=90,width=\textwidth]{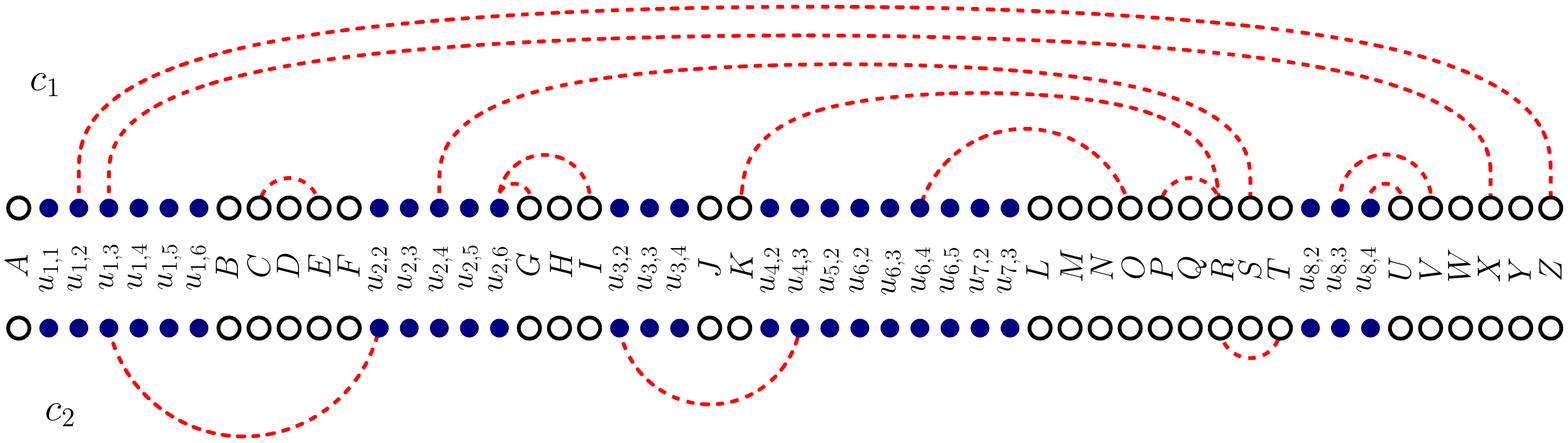}
\includegraphics[angle=90,width=\textwidth]{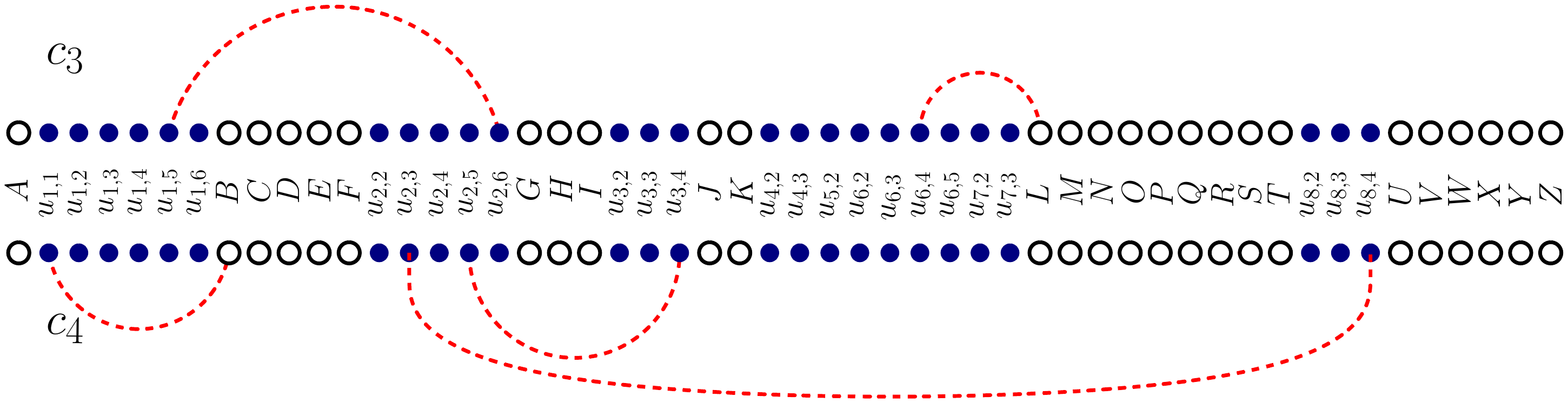}
\caption{Book embedding of the crossing edges in $X$ for the 2-level graph $H(C_i)$ in
 Fig.~\ref{fig:2-level} on the four pages $c_1$, $c_2$, $c_3$, $c_4$.}
\label{fig:embedding-crossing}
\end{figure}

\noindent
\textbf{Case D3: $(a,b)$ is a level $i+1$ edge}. In this case $(a,b)$ is crossed by a binding edge $(c,d)$,
 where one vertex (say $c$) is an outer vertex, and the other vertex (say $d$) is a cutvertex in $D$.
 The four vertices $a,b,c,d$ form a $K_4$ in $H(C_i)$ with skeleton $acbd$ whose interior is
 vertex-empty. Let $B_j$ and $B_{j'}$ be the two blocks of $D$ containing $a$ and $b$, respectively, with
 the common vertex $d$. Then, either one of $B_j$ and $B_{j'}$ is the parent of the other in $\TT$,
 where $d$ is the leader for the child block, or both $B_j$ and $B_{j'}$ are the children of a common
 parent block in $\TT$ and $d$ is the leader for both of them. In either case, assume without loss of
 generality that the dominator of $B_j$ comes before the dominator of $B_{j'}$ in the cw-order
 around $C_i$ (i.e., the dominator of $B_j$ is placed before that of $B_{j'}$ in $\Gamma$). This implies
 that the vertices of $B_j$ are all placed before the vertices of $B_{j'}$ except for its leader.
 Since $b$ is adjacent to the leader $d$ in $B_{j'}$, $b$ is either the first or the last vertex of $B_{j'}$
 (except for its leader) in $\Gamma$. We call the edge $(a,b)$, the \textit{first} (resp. \textit{last})
 crossing edge for the block $B_{j'}$. 
 Note that if $(a,b)$ is the last crossing edge for $B_{j'}$, then $c$ is the dominator for $B_{j'}$.
 We assign a crossing edge in $X$ at level $i+1$ to page $c_2$ if it is the first crossing edge
 for some block $B_{j'}$; see Fig.~\ref{fig:embedding-crossing}.
 The last crossing edges of the blocks are assigned to either page $c_3$ or page $c_4$ in Case D4.

\noindent
\textbf{Case D4: the other case: $(a,b)$ is a forbidden binding edge for some block or the last
 crossing edge for some block}.
Since the edges in $X$ do not cross each other, for each block $B_j$, there is at most on edge, which
 is either a last crossing edge or a forbidden binding edge for $B_j$. These edges are assigned to page
 $c_3$ or $c_4$ as follows. Consider the rooted block-cut tree $\TT$ for the blocks, rooted at $B_1$.
 For each block at the even (resp. odd) level of $\TT$, we assign its forbidden binding edge or last
 crossing edge (if any) to page $c_3$ (resp. $c_4$); see Fig.~\ref{fig:embedding-crossing}.

We now prove Lemma~\ref{lem:2-level} by showing that for any of the seven pages, there is no conflict
 between the edges assigned to it. This follows directly from~\cite{Yanna89} for the three pages $p_1$,
 $p_2$ and $p_3$, since the edges assigned to these three pages forms the planar graph $H'$ and
 the order of the vertices and the edge assignment on these three pages for $H'$ is exactly the same as
 in~\cite{Yanna89}. For the edges assigned to the remaining four pages $c_1$, $c_2$, $c_3$, $c_4$,
 we have the following Lemmas.

\begin{lemma}
There is no conflict between edges assigned to page $c_1$.
\label{lem:c1}
\end{lemma}
\textit{Proof.}
 The edges assigned to $c_1$ are the level $i$ edges and the binding edges in $X$, not
 incident to the leader of any block. We show that no two of them create a conflict.
\begin{wrapfigure}{r}{.42\textwidth}
\vspace{-.4cm}
\centering
\includegraphics[width=0.4\textwidth]{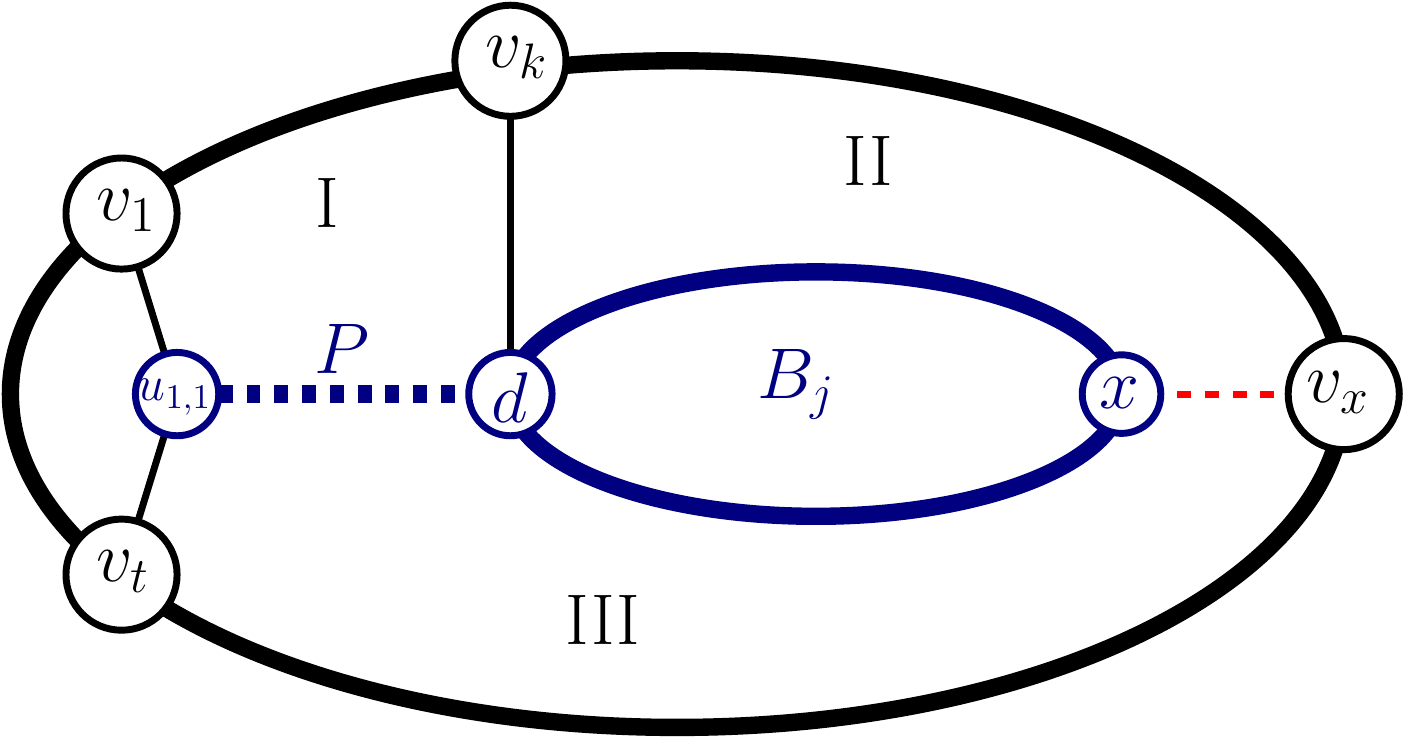}
\caption{Illustrations for the proof of Lemma~\ref{lem:c1}.
\vspace{-.4cm}
}
\label{fig:binding}
\end{wrapfigure}
Since the vertices of $C_i$ are placed in circular (clockwise) order of its boundary, and no two edges in
 $X$ crosses each other in the embedding $H(C_i)$ (only one edge from each crossing pair is taken),
 no two level $i$ edges in $X$ are in conflict with each other. Therefore it is sufficient to show that
 no binding edge in $X$ is in conflict with any other binding edge or level $i$ edge in $X$.

Consider a binding edge $(x,v_x)$ assigned to page $c_1$, where $v_x$ is an outer vertex and $x$ is an
 inner vertex; see
 Fig.~\ref{fig:binding}. Let $x$ is assigned to the block $B_j$. Let $v_k$ be the dominator of $B_j$
 and $d$ be the leader of $B_j$. Also consider a path $P$ from the first inner vertex $u_{1,1}$ to $d$ in
 the planar skeleton of $H(C_i)$ (the trivial path if $j=1$).
The block $B_j$, the two edges $(x,v_x)$ and $(d,v_k)$, along with the path $P$  and the two edges
 $(u_{1,1},v_1)$, $(u_{1,1},v_t)$ partitions the interior of $C_i$ in the following parts: (i)~the interior of
 $B_j$, (ii)~the interior of the triangle $(u_{1,1},v_1,v_t)$ and (iii)~the three regions marked by I, II and III
 in Fig.~\ref{fig:binding}.
Since the path $P$ and the boundary of $B_j$
 belongs to the planar skeleton of $H(C_i)$ and since the edge $(x,v_x)$ is a crossing edge, each edge
 assigned to page $c_1$ is embedded in the interior of one of the three regions I, II or III.

All the level $i$ vertices in region I are placed on or before $v_k$ in $\Gamma$. Since $v_k$ is the
 dominator of $B_j$, it is placed before any vertex assigned to $B_j$, including $x$. Thus any level $i$
 edge in $X$ lying in region I has both their end-vertices placed before both $x$ and $v_x$, and hence
 does not create a conflict with $(x,v_x)$. One the other hand, all level $i+1$ vertices $y$ in region I
 including the ones on $P$ are also placed before $x$. Indeed if $B_{j'}$ is the block to which $y$ is
 assigned to, then either $B_{j'}$ is dominated by an outer vertex placed before $v_k$, or $B_{j'}$ is
 dominated by $v_k$, but its vertices are placed before those of $B_j$, following the consecutive
 (or the nested) method of placement. Thus both end-vertices of any binding edge in $X$ lying in region I
 are also placed before $x$ and $v_x$ and hence does not create any conflict with $(x,v_x)$.

Again all the level $i$ vertices in region II except for $v_k$ are placed after $x$ and before $v_x$. Similarly
 all the level $i+1$ vertices in region II except for $d$ are placed on or after $x$ and before $v_x$.
 Due to the way, we select the edges in $X$, no binding edge or level $i$ edge in $X$ lying
 in region II is incident to $v_k$. Furthermore no binding edge incident to $d$ are assigned to page
 $c_1$. Thus all the binding edges and level $i$ edges assigned to page $c_1$ have end-vertices placed
 between $x$ and $v_x$; hence they create no conflict with $(x,v_x)$.

All the level $i$ vertices in region III are placed on or after $v_x$. Thus all level $i$ edges in $X$ lying in
 region III have both their end vertices placed after both $x$ and $v_x$, and hence they create no conflict
 with $(x,v_x)$. On the other hand, the level $i+1$ vertices on $P$ or on the boundary of $B_j$ lying region
 III are placed before $x$ and the binding edge incident to them does not create conflict with $(x,v_x)$.
 Finally all the level $i+1$ blocks strictly in region III are dominated by the vertices placed on or after $v_x$.
 Indeed, the only possible planar edge crossing the region boundary would have been incident to the level
 $i$ vertex $v_{x-1}$ just before $v_x$, and it would have crossed the edge $(x,v_x)$. However in that
 case, the other end vertex of such an edge would have been on a block dominated by $v_{x-1}$ and $x$
 would have been its leader, which is a contradiction since the edge $(x,v_x)$ is assigned to page $c_1$.
Thus all the binding edges in region III incident to some level $i+1$ vertex neither on $P$ nor $B_j$, have
 both the end-vertices placed after $x$ and $v_x$, and hence they do not create conflict with $(x,v_x)$. \qed

For a planar Hamiltonian graph, the order of the vertices from a Hamiltonian cycle induces a 2-page
 book embedding~\cite{bk-btg-79}. Furthermore if the graph is outerplanar, then this order of the vertices
 on the outer cycle induces a 1-page book embedding. We use these two facts to show that there is no
 conflict on the pages $c_2$, $c_3$ and $c_4$.

\begin{lemma} There is no conflict between edges assigned to the pages $c_2$, $c_3$ and $c_4$.
\label{lem:c234}
\end{lemma}
\begin{proof}
Consider a cycle $C$ defined by the vertex order in $\Gamma$; i.e., the vertices of $C$ are all the vertices
 of $H(C_i)$, and for each consecutive vertex in $\Gamma$, there is an edge in $C$, along with an edge
 between the first and the last vertex on $\Gamma$. We show that all the edges assigned to page $c_2$
 along with this cycle $C$ forms an outerplanar graph with $C$ as the outer cycle. We also show that
 all the edges assigned to pages $c_3$, $c_4$, along with $C$ forms a planar graph with the Hamiltonian
 cycle $C$. The claim thus follows.

First, consider a fixed planar embedding of $C$ induced from the embedding of $H(C_i)$. Delete all the
 edges from $H(C_i)$ except for the edges on $C_i$ and the edges on the boundary of each block $B_j$.
 For each block $B_j$, delete the edge between its leader and the last vertex (in the counterclockwise
 order). Finally also delete each edge $(v_k,v_{k+1})$ for each outer vertex $v_k$, which is a dominator
 of some block. Finally add the following edges for each dominator $v_k$. If $v_k$ dominates only
 a single block $B_j$, then add the edge between $v_k$ and the first vertex assigned to $B_j$, and the
 edge between $v_{k+1}$ and the last vertex assigned to $B_j$. These two edges can be routed
 without a crossing near the (now removed) edge between the leader and the last vertex of $B_j$; see
 Fig.~\ref{fig:hamil-c234}. If $v_k$ dominates more than one blocks $B_{j1}$, $B_{j2}$, $\ldots$, $B_{jt}$
 in this cw-order, then we add the edge from $v_k$ to the first vertex of $B_{j1}$, and an edge
 from $v_{k+1}$ to the last vertex of $B_{jt}$. Also for $1\le l<t$, add an edge from the last vertex of
 $B_{jl}$ to the first vertex of $B_{j(l+1)}$. Again all these edges can be routed near the (now removed)
 edges between the leader and the last vertex of the blocks. This gives a planar embedding of $C$.

\begin{figure}[tb]
\vspace{-0.3cm}
\centering
\includegraphics[width=0.8\textwidth]{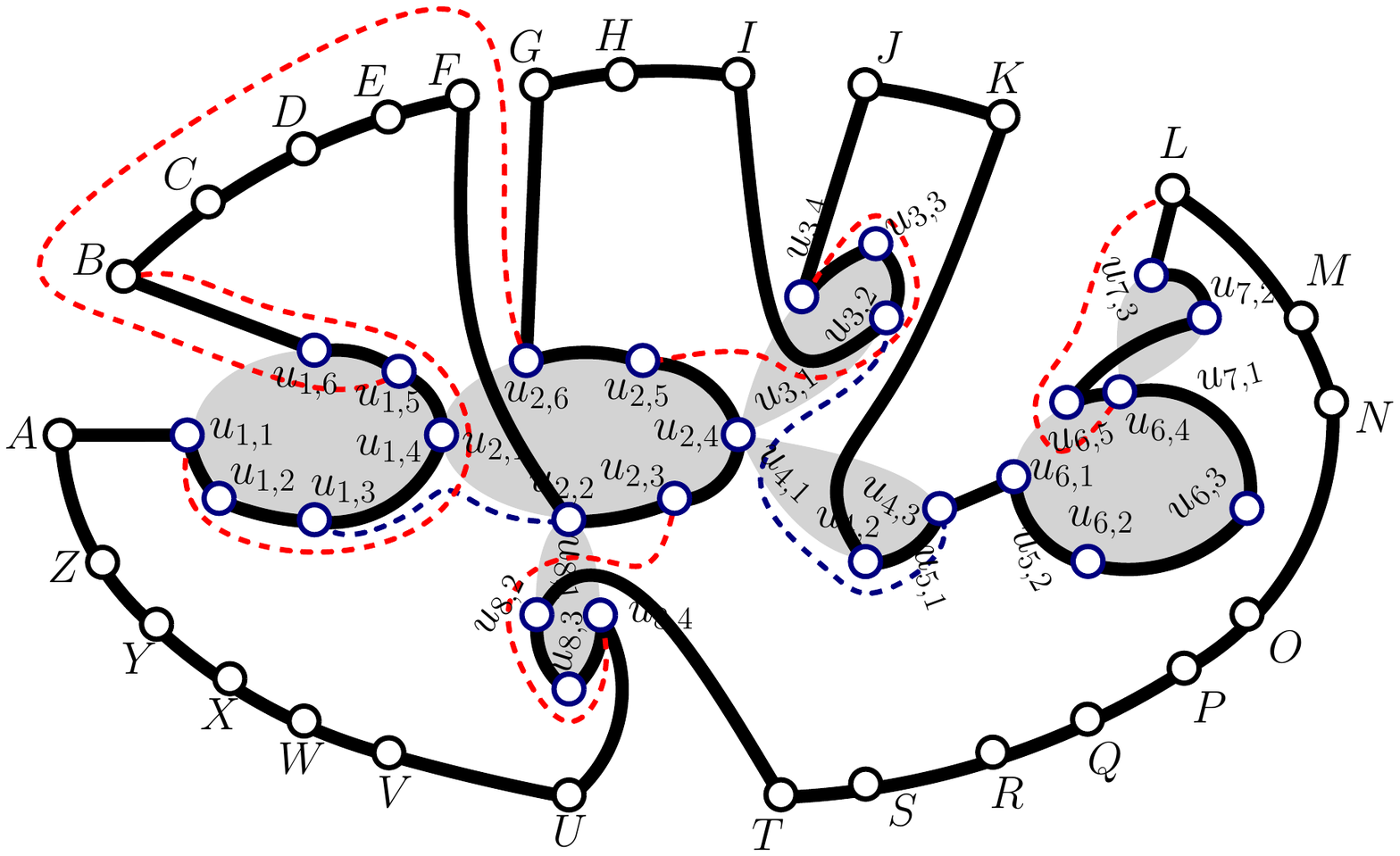}
\caption{Construction of the Hamiltonian cycle from $\Gamma$ (thick black edges).
 The blue dotted edges are the first crossing edges of the blocks. The red dotted edges are the
 last crossing edges and the forbidden binding edges for the blocks.}
\label{fig:hamil-c234}
\end{figure}

We now show that all the edge assigned to page $c_2$ can be added in the interior of $C$ without
 crossing. The edges assigned to $c_2$ are the first crossing edges of the blocks. For any block $B$,
 with leader $d$, its first crossing edge (if any) is between the first vertex $u_1$ assigned to $B$ and
 the vertex $x$ of $B'$ preceding $d$, where $B'$ is either the parent of $B$ in $\TT$ or the sibling
 of $B$ in $\TT$ just clockwise of it (Note that, in the later case, $x$ is the last vertex of $B'$).
 We route such an edge as follows. We follow the boundary of $B$ in cw-order from $u_1$ to
 $d$, then cross the boundary of $B'$ if it is a sibling of $B$. Finally we follow the boundary of $B'$
 (counterclockwise in $B'$ is the parent of $B$; clockwise otherwise) to $x$; see Fig.~\ref{fig:hamil-c234}.
 The routed edges are planar and are in the interior of $C$. Hence they induce an outerplanar embedding,
 implying that edges assigned to $c_2$ can be embedded on a single page.

Finally, the edges assigned to page $c_3$, $c_4$ are the last crossing edges and the forbidden edges
 of the blocks. We show how we route them in the embedding of $C$ without crossings.
 Consider a block $B$ with the leader $d$ and a last crossing edge $e$. Then $e$ is between the last
 vertex of $B$ and the vertex $x$ on the parent of $B$ in $\TT$ following $d$ in the counterclockwise
 order. If $B$ is at an even level in $\TT$, we route the edge outside of $C$, following the edge to its
 dominator. Then we follow the boundary of $C$ until we reach the last vertex of $B'$. Finally we follow
 the inside of the boundary of $B'$ to $x$. If $B$ is in the odd level, we route $e$ inside following the
 boundary of $B$ in the cw-order until $d$, then cross the boundary and finally follow the boundary
 of $B'$ in the ccw-order to $x$; see Fig.~\ref{fig:hamil-c234}. For each block, if its last
 crossing edge follows its outside
 boundary, then the edges from its children blocks following its inside boundary and vice versa.
 Furthermore for the children of a block $B$ in cw-order, their leaders also appear in the clockwise
 order on $B$ and the edges from each child only covers the boundary of $B$ only up to its leader.
 Thus these edge do not create crossing. Finally for a forbidden edge $e$ of a block $B$, between it
 leader and its dominator, we route $e$ in the same route for the last crossing edge; see
 Fig.~\ref{fig:hamil-c234}. Thus all these edges along with $C$ forms a planar graph with the Hamiltonian
 cycle $C$, and hence the can be embedded in the two pages $c_3$, $c_4$.
\end{proof}

\subsection{Drawing 3-Connected 1-Planar Graphs}

Here we describe a 12-page book embedding algorithm for any 3-connected 1-planar graph $G$. We
 first show how we order
 the vertices of $G$ using the vertex placement order for 2-level subgraphs from the previous section. We
 then show how we assign the edges of $G$ into a small number of pages.

As we described in the previous Section, we may assume that $G$ is a normal planar-maximal 1-planar
 graph. We use a ``peeling'' technique to find a linear order for the vertices of the graph $G$ level-by-level
 using the algorithm for Lemma~\ref{lem:2-level}. We first find and order of the vertices on the outer cycle
 $C_0$ (level 0 vertices) such that the vertices are placed in the cw-order around $C_0$. We then
 traverse the graph outside in and iteratively use the algorithm for Lemma~\ref{lem:2-level} to place the
 internal vertices. For the 2-level graphs between levels $i$ and $i+1$, we consider that the vertices of
 level $i$ have already been placed and we place the vertices of level $i+1$ using the algorithm for
 Lemma~\ref{lem:2-level}. 
 
 Consider a 2-level graph $H(C_i)$ between levels $i$ and $i+1$, where
 $C_i=\langle v_1, \ldots v_t\rangle$ is the outer boundary of $H(C_i)$. If the cycle $C_i$ is the first block
 in a 2-level graph between levels $i-1$ and $i$, then the interval between the vertices of $C_i$ does not
 contain any other vertex and we can use the algorithm in the previous section to place the level $i+1$
 vertices inside $H(C_i)$ between the already placed vertices of $C_i$. Otherwise there is
 some vertex of level $j<i$ between $v_1$ and $v_2$, but the remaining vertices ($v_2,\ldots v_t$) are
 in a consecutive interval. 
 In this case we again place the level $i+1$ vertices
 inside $H(C_i)$ as in the algorithm for Lemma~\ref{lem:2-level}, but we place the vertices of level $i+1$ blocks
 dominated by $v_1$, just before $v_2$ (after all possible vertices
 of level $j<i$). In either case, the vertices on each level $i+2$ cycles are placed in an interval with no
 vertices of level $j\le i$ in between. Call this Algorithm \textbf{Order-Vertices}. We thus have the following
 lemma, whose proof follows from the above discussion; also see~\cite{Yanna89}:

\begin{lemma} Let $\Gamma$ be the vertex order for a normal planar-maximal 1-planar graph $G$,
obtained by Algorithm \textbf{Order-Vertices}. Let $C_i$ be some level $i$ cycle in $G$. Then all vertices
 at level $i+1$ inside $C_i$ are placed strictly between two consecutive level $i+1$ vertices $v_j$ and
 $v_{j'}$ in $\Gamma$.
\label{lem:independent}
\end{lemma}

Lemma~\ref{lem:independent} implies that with this vertex order, no level $i+1$ edge of $G$ conflicts with
 any level $j$ edge with $j<i$. We thus can iteratively use the drawing algorithm in Lemma~\ref{lem:2-level}
 to obtain a book embedding of $G$ as follows:

\begin{theorem}
\label{th:14-page}
Every 3-connected 1-planar graph $G$ has a book embedding on 14 pages.
\end{theorem}
\begin{proof}
Let $G$ be a normal planar-maximal 1-planar graph. Using Algorithm
 \textbf{Order-Vertices} we find a linear order of the vertices in $G$. We now again use the ``peeling''
 technique to embed the edges of $G$ level-by-level following the algorithm for Lemma~\ref{lem:2-level}.
 Let $p_1, \ldots p_6$, $c_1, \ldots, c_8$ denote the 14 pages.
 We first embed the outer cycle $C_0$ (level 0 vertices) in a single page (page $p_1$). Then
 for each 2-level graph between levels $i$ and $i+1$, we iteratively use the pages $p_1$, $p_2$, $p_3$,
 and $c_1$, $c_2$, $c_3$, $c_4$ to embed all the edges, when $i$ is even; and we use the pages
 $p_4$, $p_5$, $p_6$, and $c_5$, $c_6$, $c_7$, $c_8$ when $i$ is odd. By Lemma~\ref{lem:2-level},
 each 2-level subgraph is drawn without conflict, and by Lemma~\ref{lem:independent}, the
 edge in any 2-level does not create conflict with any 2-level subgraph in a deeper level.
\end{proof}

We can actually reduce the number of pages a little.

\begin{theorem}
\label{th:12-page}
Every 3-connected 1-planar graph $G$ has a book embedding on 12 pages.
\end{theorem}
\begin{proof} We can obtain a book embedding of $G$ on 12 pages as a corollary of the construction
 in Theorem~\ref{th:14-page} after a post processing step. We note that all the 2-level planar graphs $H'$
 at all levels $i$ of $G$, together induce a planar subgraph $H$ of $G$, and are embedded on the six
 pages $p_1, \ldots p_6$. Furthermore the order of the vertices in this book embedding is the same as the
 one obtained by the algorithm by Yannakakis~\cite{Yanna89} for a book embedding of $H$. Thus we use
 the algorithm by Yannakakis~\cite{Yanna89} to embed $H$ on only four pages, resulting in a total of 12
 pages.
\end{proof}

\section{Book Embedding of General 1-Planar Graphs}

For the general case we may assume that the input graph is a planar-maximal graph and hence
 is 2-connected. We first extend the procedure of the \textit{normalization} to the case of a planar-maximal
 1-planar graph $G$. A pair of vertices $\{u, v\}$ of $G$ share more than two crossing edge
 pairs if and only if $\{u, v\}$ form a separation pair in $G$~\cite{ABK13}. During the normalization,
 for any separation pair $\{u, v\}$, we route the edge $(u,v)$ such that all the crossing edge pairs
 with $u$, $v$ as end-vertices falls on the same side of $(u,v)$; see Fig.~\ref{fig:separation}.

Suppose there is a separation pair $\{u, v\}$, with a decomposition $G - \{u, v\} = \{H_0, \ldots, H_k\}$
 for some  $k \geq 1$. For any such component $H_j$, let $H^*_j$ be the subgraph of $G$ induced
 by the vertices of $H_j$ and $\{u,v\}$. Then for at most one component $H_j$, $u$ and $v$ are not
 on the outerface of $H^*_j$. Assume thus without loss of generality that $H^*_1$, $\ldots$, $H^*_k$
 all have $u$, $v$ on the outerface. We call $H_0$ the \emph{main} component and $H_1$, $\ldots$,
 $H_k$ the \textit{inner components} for $\{u,v\}$. Also call $H^*_1$, $\ldots$, $H^*_k$ the
 \textit{extended inner components}.
 The edge $(u,v)$ is called \emph{separating edge}.
Note that the inner components can be permuted
 and flipped at $\{u,v\}$.
 In a normalized planar maximal embedding $\mathcal{E}(G)$ of $G$, the inner components $H_1$,
 $\ldots$, $H_k$ are attached to $(u,v)$ and are embedded on one side of $(u,v)$, say in this
 ccw-order at $u$.
The components are separated by one or two pairs of crossing edges; see Fig.~\ref{fig:separation},
 and they may also be separated by copies of the separation
 edge~\cite{Bran-vis-14,Bran-map-15}. The embeddings of the extended inner components
 are $B$- or $W$-configurations, defined by \cite{thomassen88}, and hence the boundaries of
 the inner components are triangles and quadrangles.
\begin{wrapfigure}{r}{.39\textwidth}
\centering
\includegraphics[width=0.36\textwidth]{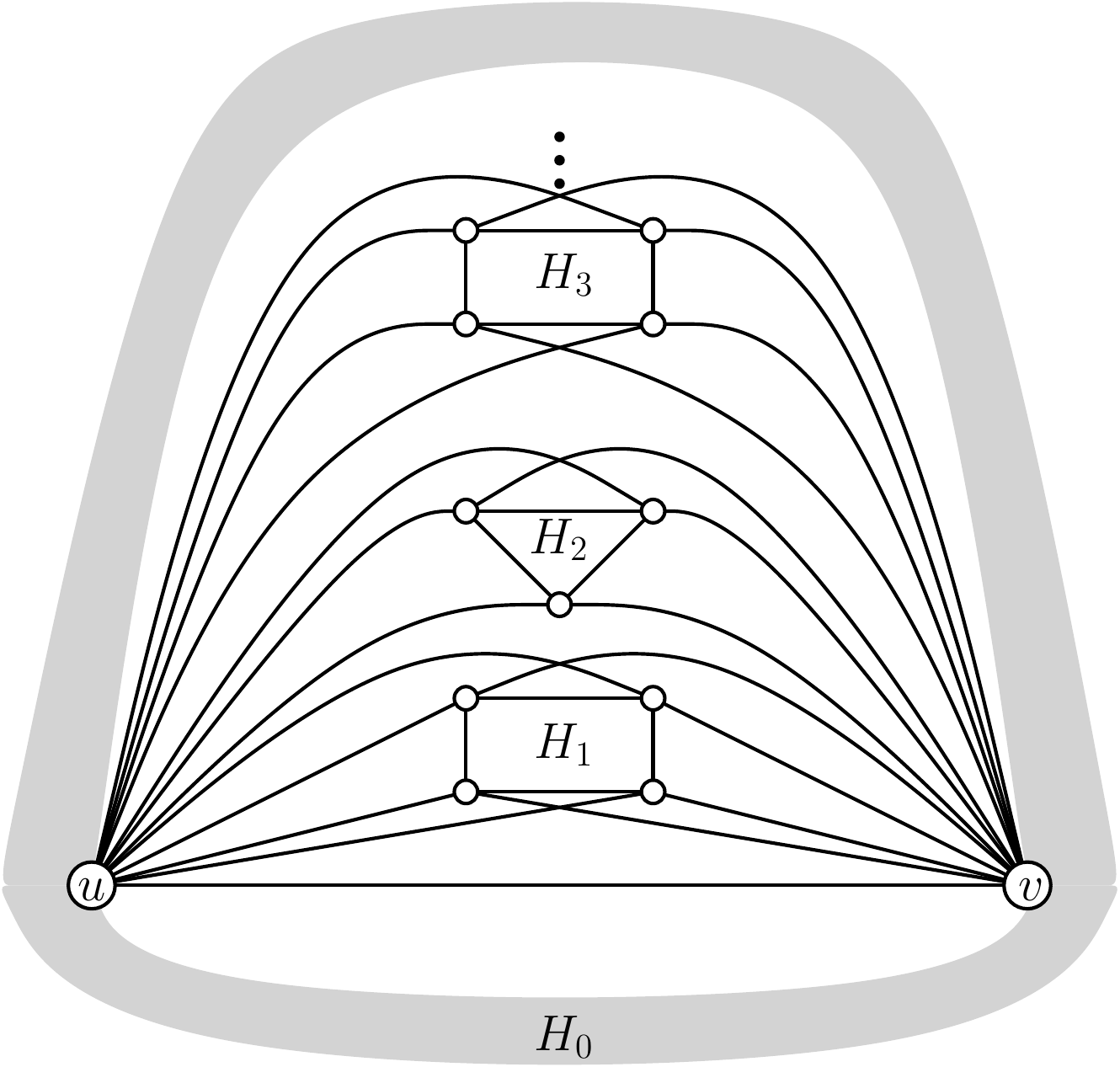}
\caption{A separation pair and the corresponding components.
\vspace{-0.8cm}
}
\label{fig:separation}
\end{wrapfigure}

We now extend our 14-page book embedding of 3-connected 1-planar graphs and the ``peeling technique''
 from Section~\ref{sec:3-connected}.

\begin{theorem} \label{2connected}
Every 1-planar graph $G$ has a book embedding on 16 pages.
\end{theorem}

\begin{proof}
We proceed as in the case of 3-connected graphs. However we extend the peeling technique here to
 deal with the
 inner components for the separation pairs. Let the \textit{main graph} $G_0$ be obtained from $G$
 by deleting all the inner components for all the separation pairs. Clearly $G_0$ is 3-connected.
 For each separation pair $\{u, v\}$, the edge $(u,v)$ is a planar edge and if $(u,v)$ is an edge of the
 main graph, then by the peeling technique, $u$, $v$ are on the same level or on consecutive levels.
 Let $H_1$, $\ldots$, $H_k$ be the inner components for $u, v$.
 We then assign the vertices on the outer boundary $O_j$ for each inner component $H_j$ on the higher
 (i.e., deeper)
 of the two levels for $u$ and $v$. For the remaining vertices of $H_j$ we proceed with the
 peeling technique recursively and assign them to subsequent levels. Let $u$, $v$ belong to some
 2-level subgraph $H(C_i)$ of the main graph. Then the vertices on the outerboundary for each inner
 component for $u,v$ and the edges between these outer vertices vertices and $u$, $v$ are on the
 2-level subgraph for $G$. We now show how we place these vertices and assign the edges to augment
 the book embedding $\Gamma$ of $G_0$. In addition to the 14 pages used in $\Gamma$, we use two
 more pages $q_1$ and $q_2$ for 2-level subgraphs at odd and even levels, respectively.

For each separating edge $(u, v)$ on the main graph, with $u$ placed before $v$ in $\Gamma$,
 insert the vertices on the outerboundary of each inner component for $u,v$ consecutively,
 to the immediate left of $v$ (in cw-order if $v$ is on odd level and in ccw-order
 otherwise). If there is more than one inner component for $u,v$, the order of their placement is arbitrary.
If several separating edges are incident to $v$, with the other end-vertex, say $w_1$, $\ldots$,
 $w_q$, all placed before $v$ and in this order in $\Gamma$, insert the vertices of the corresponding
 inner components in reverse order (i.e., the inner components for $w_q$, $\ldots$, the inner components
 for $w_1$).

The edges on the outerboundary are assigned to $c_1$ or $c_5$ for odd and even levels,
 respectively; they do not create conflicts because they form simple cycles of length 3 or 4 and the
 vertices are consecutive. For each inner component $H_j$ for separation pair $\{u,v\}$, the edges
 from $u$ to the vertices on $O_j$ are assigned to the same page as $(u, v)$, and the edges from
 $v$ to the vertices of $O_j$ are assigned to page $q_1$ (resp. $q_2$) for odd (resp. even) levels.
 Here the edges to $v$ do not create conflicts with each other since they are all incident to $v$, and
 they do not create conflicts with other edges  on $q_1$ (or $q_2$) since they are all placed immediately
 before $v$. Similarly the edges to $u$ do not cross each other since they are all incident to $u$ and
 they do not create conflicts with other edges in the same page since they follow the planar edge $(u,v)$
 assigned to the same page.

We recursively place the vertices inside each inner component during the computation for 2-level
 subgraphs on subsequent levels. Since we assign edges from 2-level subgraphs at odd and even
 levels on disjoint pages, following the argument of Lemma~\ref{lem:independent} the edges assigned
 to each of the 16 pages
 do not create conflicts.
\end{proof}

It is $NP$-hard to determine whether a planar graph (which is a
subclass of 1-planar graph) is sub-Hamiltonian. Hence, the minimum
number of pages of a 1-planar graph cannot be computed efficiently.
However, our algorithm takes only linear time, given a
1-planar embedding.

\begin{theorem}
\label{th:time}
There is a linear time algorithm to construct book embedding of a  general and a 3-connected 1-planar
 graph on 16 and 12 pages, respectively, given a 1-planar embedding.
\end{theorem}
\begin{proof}
Given the 1-planar embedding, the normal planar maximal augmentation can be obtained in linear time.
The crossing edges to be removed are selected in constant time per edge. Yannakakis algorithm for
 planar graphs runs in linear time, and the assignment of a removed edge to a page takes constant time
 per edge. Since there are at most $4n-8$ edges, the algorithm runs in linear time.
\end{proof}

If the input graph is planar and Hamiltonian, the order of the vertices from a Hamiltonian cycle induces a 2-page book embedding~\cite{bk-btg-79}. We can use this as follows.

\begin{corollary}
A 1-planar graph $G$ has a 4-page book embedding if the planar skeleton is Hamiltonian.
\end{corollary}

\begin{proof} Let $\mathcal{P}(G)$ be the planar skeleton of $G$ with Hamiltonian cycle $C$.
For each pair $(a,b)$ and $(c,d)$ of crossing edges assign $(a,b)$ to a set $X_1$ and $(c,d)$ to $X_2$
 arbitrarily. By slight the abuse of notation, denote with $X_1$ ($X_2$) the subgraphs of $G$ induced
 by $X_1$ ($X_2$).
 Both $G_1=\mathcal{P}(G)\cup X_1$ and $G_2=\mathcal{P}(G)\cup X_2$ contain
  Hamiltonian cycle $C$. Using the linear order of $C$ we can embed $G_1$ in 2 pages and
 $G_2$ in 2-pages, yielding a book embedding for $G$ on 4-pages with duplicate edges
 of $\mathcal{P}(G)$ removed.
\end{proof}

\section{Conclusion}
We showed that general and 3-connected 1-planar  graphs have a book embedding on 16 and 12 pages,
 respectively, and the book embedding can be computed in linear time from a given 1-planar embedding.
 Our bound improves upon the bound of 39 given by Bekos~\textit{et al.}~\cite{BBKR15}.
 The extended wheel graphs $XW_{2k}$ for $k=4,5,6$ require 4 pages; see Fig.~\ref{fig:counter}. This was shown using a program which exhaustive searches all possible vertex orders and assignments of edges to pages. The natural open problem is to close the gap between the lower and upper bounds.
 Specifically, are there 1-planar graphs
\begin{wrapfigure}{r}{.32\textwidth}
\centering
\includegraphics[width=0.25\textwidth]{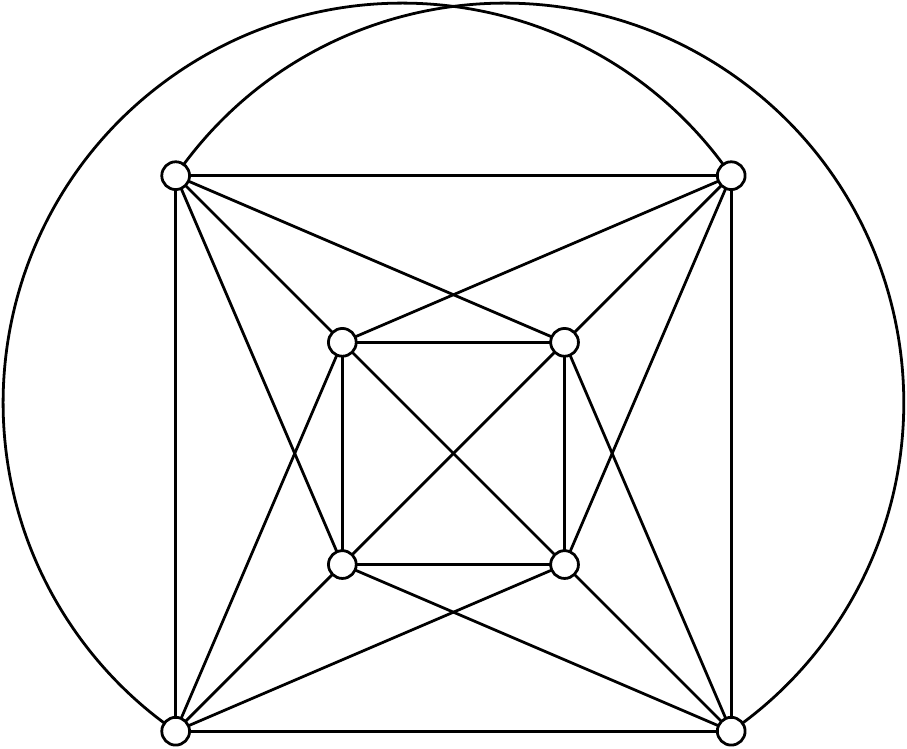}
\caption{The extended wheel graph $XW_8$ requires 4 pages.
}
\label{fig:counter}
\end{wrapfigure}
 that require even 5 pages?
 What is the lowest number of pages that suffices for 1-planar graphs, or 3-connected 1-planar graphs?
These questions mirror the remaining big open problem for planar graphs: are there planar graphs that require 4 pages, or are all planar graphs embeddable on 3 pages?

\subsubsection*{Acknowledgments.}
We thank O.~Aichholzer, M.~Bekos, D.~Eppstein,  K.~Hanauer, S.~Pupyrev, A.~Schulz,
 and A.~Wolff for useful discussions and in particular M.~Bekos and M.~Kaufmann for pointing out
 an error in an earlier version of this paper.

\bibliographystyle{abbrv}
{
\begin{small}

\end{small}
}

\end{document}